\documentclass[runningheads]{llncs}
\usepackage{amssymb} % symbols for natural numbers
\usepackage{amsfonts} % symbols for natural numbers
\usepackage{amsmath} % own operators
\usepackage{latexsym}
\usepackage{enumerate}
\usepackage{graphicx} % "\includegraphics" command with width argument
\usepackage{psfrag}
\usepackage{comment}
\usepackage{bold-extra}
\usepackage{url}
\usepackage{mathrsfs}

\newcommand{\Hlineny}{Hlin\v{e}n\'{y}}
\newcommand{\EF}{Ehrenfeucht-Fra\"{i}ss\'{e}}

\newcommand{\rank}[1]{\ensuremath{{\rm rank}(#1)}}
\newcommand{\lab}[2]{\ensuremath{{\it lab}\colon #1 \to #2}}
\newcommand{\Lab}{\ensuremath{{{\it lab}}}}
\newcommand{\GF}[1]{\ensuremath{{\rm GF}(#1)}}

\newcommand{\myotimes}[3]{\ensuremath{\mathrel{{\otimes}[#1|#2,#3]}}}
\newcommand{\moreotimes}[1]{\ensuremath{\otimes_{#1}}}
\newcommand{\vect}[1]{\ensuremath{\mathbf{#1}}}
\newcommand{\mso}{\ensuremath{{\rm MSO}}}
\newcommand{\bem}[1]{{{\em #1}}}
\newcommand{\mcg}{\ensuremath{\mathcal{MC}}}
\newcommand{\mycart}[3]{\ensuremath{\mathrel{{\times}(#1, #2, #3)}}}

\def\root{\ensuremath{{\rm root}}}
\def\children{\ensuremath{{\rm children}}}
\def\subtree{\ensuremath{{\rm subtree}}}
\def\subtrees{\ensuremath{{\rm subtrees}}}
\def\free{\ensuremath{{\rm free}}}
\def\ordind{\ensuremath{{\rm Ord}}}
\def\indicator{\ensuremath{{\rm ind}}}

\newcommand{\ms}[1]{\ensuremath{\mathscr{#1}}}
\newcommand{\qr}[1]{\ensuremath{{\rm qr}(#1)}}

\newcommand{\msoeq}[1]{\ensuremath{\equiv_{#1}^{{\rm MSO}}}}
\newcommand{\dom}[1]{\ensuremath{{\rm domain}(#1)}}
\newcommand{\range}[1]{\ensuremath{{\rm range}(#1)}}

\newcommand{\fullchar}[4]{\ensuremath{{\rm FC}_{#1}(#2,#3,#4)}}
\newcommand{\fc}[2]{\ensuremath{{\rm FC}_{#1}(#2)}}
\newcommand{\redchar}[4]{\ensuremath{{\rm RC}_{#1}(#2,#3,#4)}}
\newcommand{\rc}[2]{\ensuremath{{\rm RC}_{#1}(#2)}}

\newcommand{\eqv}[1]{\ensuremath{\equiv_{#1}}}

\newcommand{\msoone}{\ensuremath{{\rm MSO}_1}}

\newcommand{\opt}{\ensuremath{\rm opt}}

\newtheorem{prop}{Proposition}
\spnewtheorem*{main}{The Main Theorem}{\bfseries \upshape}{\itshape}
\begin{document}
\titlerunning{Linear-Time Algorithms for Graphs of Bounded Rankwidth}
\title{Linear-Time Algorithms for Graphs of Bounded Rankwidth:
A Fresh Look Using Game Theory 
\thanks{This work is supported by the Deutsche Forschungsgemeinschaft (DFG).}} 
\author{Alexander Langer \and Peter Rossmanith \and Somnath Sikdar}
\institute{RWTH Aachen University, 52074 Aachen, Germany.}
\date{}

\maketitle

\begin{abstract}
We present an alternative proof of a theorem by Courcelle, Makowski and Rotics~\cite{CMR00}
which states that problems expressible in \msoone\ are solvable in linear time 
for graphs of bounded rankwidth. Our proof uses a game-theoretic approach 
and has the advantage of being self-contained, intuitive, and fairly easy to follow.
In particular, our presentation does not assume any background in logic or automata theory.
We believe that it is good to have alternative proofs of this important result.
Moreover our approach can be generalized to prove other results of a similar flavor,
for example, that of Courcelle's Theorem for treewidth~\cite{Cou90,KLR10}.  
\end{abstract}

\section{Introduction}
In this paper we give an alternate proof of the theorem by Courcelle, Makowski and Rotics~\cite{CMR00}:
\emph{Every decision or optimization problem expressible in \msoone\ is linear time solvable on graphs of bounded
cliquewidth}. 
We prove the same theorem for graphs of bounded rankwidth. Since rankwidth 
and cliquewidth are equivalent width measures in the sense that a graph has bounded rankwidth
iff it has bounded cliquewidth, it does not matter which of these width measures is used 
to state the theorem~\cite{OS06b}.

The proof by Courcelle et al.~\cite{CMR00,CMR01} makes use of the Feferman-Vaught
Theorem~\cite{FV59} adapted to MSO (cf.~\cite{Gur79,Gur85}) and MSO transductions (cf., \cite{Cou94}).
Understanding this proof requires a reasonable background in logic and as such this proof
is out of reach of many practicing algorithmists. 
An alternative proof of this theorem has been recently published by Ganian and \Hlineny~\cite{GH10}
who use an automata-theoretic approach to prove the theorem. Our approach to 
proving this theorem is game-theoretic and what distinguishes our proof is that it is 
fairly simple, intuitive and, more importantly, self-contained. An outline of our approach 
follows.

It is known that any graph of rankwidth~$t$ can be represented
by a $t$-labeled parse tree~\cite{GH10}. 
Given any integer~$q$, one can define an equivalence relation on the class of all
$t$-labeled graphs as follows: $t$-labeled graphs~$G_1$ and~$G_2$ are 
equivalent, denoted $G_1 \msoeq{q} G_2$, iff
for every \msoone-formula of quantifier 
rank at most~$q$ $G_1 \models \varphi$ iff $G_2 \models \varphi$, i.e.,
no formula with at most~$q$ nested quantifiers can distinguish them. 
The number of equivalence classes depends on the quantifier rank~$q$ and the number of
labels~$t$ and each equivalence class can be represented by a tree-like
structure of size~$f(q,t)$, where~$f$ is a computable function of~$q$
and~$t$ only. 

This tree-like representative of an equivalence class, 
called a \emph{reduced characteristic tree of
depth~$q$} and denoted by~$\rc{q}{G}$, captures all model-checking games that can be played on graphs in that
equivalence class and formulas of quantifier rank at most~$q$. One
can construct a reduced characteristic tree of depth~$q$ given a
$t$-labeled parse tree of an $n$-vertex graph in time~$O(f'(q,t) \cdot n)$. Finally
to decide whether~$G \models \varphi$, for some
\msoone-formula~$\varphi$ of quantifier rank at most~$q$, we simply
simulate the model checking game on~$\varphi$ and~$G$ using~$\rc{q}{G}$.
This takes an additional 
$O(f(q,t))$ time and shows that one can decide whether~$G
\models \varphi$ in time~$O(f''(q,t) \cdot n)$ proving the
theorem. The notions of $q$-equivalence $\msoeq{q}$ and related two-player pebble games
(such as the \EF\ game) are fundamental to finite model theory and can be found in any book on the subject 
(cf.~\cite{EF99}). However for understanding this paper, one does
not need any prior knowledge of these concepts.

The rest of the paper is organized as follows. 
Section~\ref{sec:rankwidth_intro} 
recaps the basic definitions and properties of rankwidth. 
Section~\ref{sec:msol} is a brief introduction to monadic second order for
those who wish to see it, and has been included to make the paper self-contained.  
In Section~\ref{sec:q_equiv}, we introduce the equivalence relation~$\msoeq{q}$, 
model-checking games and characteristic trees of depth~$q$. In this section we prove
that reduced characteristic trees of depth~$q$ for $t$-labeled graphs 
indeed characterize the equivalence relation~$\msoeq{q}$ on the class 
of all $t$-labeled graphs, and that they have size at most~$f(q,t)$, for 
some computable function of~$q$ and~$t$ alone. 
In Section~\ref{sec:constructing_char_trees} 
we show how to construct reduced characteristic trees of depth~$q$ for an $n$-vertex
graph given its $t$-labeled parse tree decomposition in time~$O(f'(q,t) \cdot n)$. 
We then use all the
ingredients to prove the main theorem. We conclude in Section~\ref{sec:conclusion}
with a brief discussion of this approach and how it can be used to obtain other results.  

\section{Rankwidth: Definitions and Basic Properties}
\label{sec:rankwidth_intro}
Rankwidth is a complexity measure of decomposing a graph into a 
tree-structure known as a \emph{rank-decomposition} and was 
introduced by Oum and Seymour to study cliquewidth~\cite{OS06b}. 
Their main objective was to find an algorithm that, given as
input a graph~$G$ and an integer~$k$, decides
whether~$G$ has cliquewidth at most~$k$ in time~$O(f(k) \cdot |V(G)|^{O(1)})$.
This question is still open but Oum and Seymour showed that
rankwidth and cliquewidth are equivalent width measures 
in the sense that a graph
has bounded rankwidth if and only if it has bounded cliquewidth.
The relationship between rankwidth and cliquewidth can be expressed
by the following inequality:
\[\mbox{rankwidth} \leq \mbox{cliquewidth} \leq 2^{1 + \mbox{\scriptsize rankwidth}} -1.\]
Moreover they also showed that there does indeed exist an 
algorithm that decides whether a graph~$G$ has rankwidth at most~$k$ 
in time~$O(f(k) \cdot |V(G)|^3)$.

We shall briefly recap the basic definitions and properties of 
rankwidth. The presentation follows~\cite{GH10,Oum05}. 
To define rankwidth, it is advantageous to first consider
the notion of branchwidth since rankwidth is usually defined
in terms of branchwidth. 

\medskip

\noindent \emph{Branchwidth.} Let~$X$ be a finite set and let~$\lambda$ 
be an integer-valued function on the subsets of~$X$. We say that the
function~$\lambda$ is \emph{symmetric} if for all~$Y \subseteq X$ 
we have~$\lambda(Y) = \lambda(X \setminus Y)$. 
A \emph{branch-decomposition} of~$\lambda$ is a pair~$(T,\mu)$, where~$T$ is a
subcubic tree (a tree with degree at most three) 
and~$\mu\colon X \to \{\,t | \text{$t$ is a leaf of~$T$}\,\}$. For an 
edge~$e$ of~$T$, the connected components of~$T\setminus e$ partition
the set of leaves of~$T$ into disjoint sets~$X_1$ and~$X_2$. The \emph{width} of
the edge~$e$ of the branch-decomposition~$(T,\mu)$ is~$\lambda(\mu^{-1}(X_1))$.
The \emph{width} of~$(T,\mu)$ is the maximum width over all edges of~$T$.
The \emph{branchwidth} of~$\lambda$ is the minimum width of all 
branch-decompositions of~$\lambda$.

The branchwidth of a graph~$G$, for instance, is defined by letting~$X = E(G)$ 
and~$\lambda(Y)$ to be the number of vertices that are incident to an edge
in~$Y$ and in~$E(G)\setminus Y$ in the above definition.

\medskip

\noindent \emph{Rankwidth.} Given a graph~$G=(V,E)$ and a
bipartition~$(Y_1,Y_2)$ of the vertex set~$V$, define a 
binary matrix~$A[Y_1,Y_2]$ with rows indexed by the vertices 
in~$Y_1$ and columns indexed by the vertices in~$Y_2$ as follows: 
the $(u,v)$th entry of~$A[Y_1,Y_2]$ is~1 if and only if~$\{u,v\} \in E$. 
The \emph{cut-rank} function of~$G$ 
is then defined as the function~$\rho\colon 2^V \to \mathbf{Z}$
such that for all~$Y \subseteq V$
\[\rho(Y) = \rank{A[Y,V\setminus Y]}.\]  
A \emph{rank-decomposition} of~$G$ is a branch-decomposition of
the cut-rank function on~$V(G)$ and the \emph{rankwidth} of~$G$ is
the branch-width of the cut-rank function.

An important result concerning rankwidth is that there is an FPT-algorithm
that constructs a width-$k$ rank-decomposition of a graph~$G$, if there 
exists one, in time~$O(n^3)$ for a fixed value of~$k$.
\begin{theorem}[\cite{HO08}] Let~$k$ be a constant and~$n \geq 2$.
Given an $n$-vertex graph~$G$, one can either construct a rank-decomposition
of~$G$ of width at most~$k$ or confirm that the rankwidth of~$G$ is larger
than~$k$ in time~$O(n^3)$.
\end{theorem}

\subsection{Rankwidth and Parse Tree Decompositions}
\label{subsec:rankwidth_parse_trees}
The definition of rankwidth in terms of branchwidth is the one 
that was originally proposed by Oum and Seymour in~\cite{os06}. 
It is simple and it allows one to prove several properties of rankwidth 
including the fact that rankwidth and cliquewidth are, in fact, equivalent width
measures in the sense that a graph has bounded rankwidth if and only if it has
bounded cliquewidth. However this definition is not very useful from
an algorithmic point-of-view and this prompted Courcelle and 
Kant\'{e}~\cite{ck07} to introduce the notion of bilinear products of multi-colored
graphs and algebraic expressions over these products as an equivalent
description of rankwidth. Ganian and \Hlineny~\cite{GH10} formulated the same ideas
in terms of labeling joins and parse trees which we briefly describe here.

\medskip

\noindent \emph{$t$-labeled graphs.} A \emph{$t$-labeling} of a graph~$G$
is a mapping~$\lab{V(G)}{2^{[t]}}$ which assigns to each vertex of~$G$ a subset 
of~$[t] = \{1, \ldots, t\}$. A \emph{$t$-labeled graph} is a pair~$(G,\Lab)$, 
where~$\Lab$ is a labeling of~$G$ and is denoted by~$\bar{G}$. 
Since a $t$-labeling function may assign
the empty label to each vertex, an unlabeled graph is considered to be 
a $t$-labeled graph for all~$t \ge 1$. A $t$-labeling of~$G$ may also be 
interpreted as a mapping from~$V(G)$ to the $t$-dimensional binary vector   
space~$\GF{2^t}$ by associating the subset~$X \subseteq [t]$ with the $t$-bit
vector~$\vect{x} = x_1 \ldots x_t$, where~$x_i = 1$ if and only if~$i \in X$.
Thus one can represent a $t$-labeling~$\Lab$ of an $n$-vertex graph as 
an $n \times t$ binary matrix.
This interpretation will prove useful later on when $t$-joins are discussed.

A \emph{$t$-relabeling} is a mapping~$f\colon [t] \to 2^{[t]}$.
One can
also view a $t$-relabeling as a linear transformation from the space~$\GF{2^t}$ 
to itself and one can therefore represent a $t$-relabeling  
by a $t \times t$ binary matrix~$T_f$. For a $t$-labeled graph~$\bar{G} = (G, \Lab)$, we
define~$f(\bar{G})$ to be the $t$-labeled graph~$(G, f \circ \Lab)$, 
where~$(f \circ \Lab)(v)$ is the vector in~$\GF{2^t}$ obtained by 
applying the linear transformation~$f$ to the vector~$\Lab(v)$. It is easy
to see that the labeling~$\Lab' = f \circ \Lab$ is the matrix product~$\Lab \times T_f$.
Informally, to calculate~$(f \circ \Lab)(v)$, apply the map~$f$ to each element
of~$\Lab(v)$ and ``sum the elements modulo~2''.

We now define three operators on $t$-labeled graphs that will be used
to define parse tree decompositions of $t$-labeled graphs. These operators were
first described by Ganian and \Hlineny\ in~\cite{GH10}. The first operator 
is denoted~$\odot$ and represents a nullary operator that creates a new graph vertex
with the label~1. The second operator is the $t$-labeled join and is
defined as follows. Let~$\bar{G}_1=(G_1, \Lab_1)$ and~$\bar{G}_2=(G_2,\Lab_2)$
be $t$-labeled graphs. The \emph{$t$-labeled join} of~$\bar{G}_1$ and~$\bar{G}_2$,
denoted~$\bar{G}_1 \otimes \bar{G}_2$, is defined as taking the disjoint union 
of~$G_1$ and~$G_2$ and adding all edges between vertices~$u \in V(G_1)$ 
and~$v \in V(G_2)$ such that~$|\Lab_1(u) \cap \Lab_2(v)|$ is odd. The resulting
graph is unlabeled. 

Note that~$|\Lab_1(u) \cap \Lab_2(v)|$ is odd if and 
only if the scalar product~$\Lab_1(u) \bullet \Lab_2(v) =1$, that is, 
the vectors~$\Lab_1(u)$ and~$\Lab_2(v)$ are \emph{not orthogonal} in the
space~$\GF{2^t}$. 
For~$X \subseteq V(G_1)$, the set of vectors~$\gamma(\bar{G}_1,X) = \{\,\Lab_1(u) \mid u \in X\,\}$
generates a subspace~$\langle \gamma(\bar{G}_1,X) \rangle$ of~$\GF{2^t}$.
The following result shows which pair of vertex subsets do not generate 
edges in a $t$-labeled join operation.
\begin{prop}[\cite{GHO09}]
Let~$X \subseteq V(G_1)$ and~$Y \subseteq V(G_2)$ be arbitrary non\-empty 
subsets of $t$-labeled graphs~$\bar{G}_1$ and~$\bar{G}_2$. In the join
graph~$\bar{G}_1 \otimes \bar{G}_2$ there is no edge between any vertex
of~$X$ and a vertex of~$Y$ if and only if the 
subspaces~$\langle \gamma(\bar{G}_1,X) \rangle$ and~$\langle \gamma(\bar{G}_2,Y) \rangle$
are orthogonal in the vector space~$\GF{2^t}$.
\end{prop}

The third operator is called the $t$-labeled composition and is
defined using the $t$-labeled join and $t$-relabelings. Given
three $t$-relabelings $g,f_1,f_2\colon [t] \to 2^{[t]}$, 
the \emph{$t$-labeled composition} $\myotimes{g}{f_1}{f_2}$ is 
defined on a pair of $t$-labeled graphs~$\bar{G}_1 =(G_1, \Lab_1)$ 
and~$\bar{G}_2 =(G_2,\Lab_2)$ as follows: 
\[\bar{G}_1 \myotimes{g}{f_1}{f_2} \bar{G}_2 := \bar{H} = (\bar{G}_1 \otimes g(\bar{G}_2),\Lab),\]
where~$\Lab(v) = f_i \circ \Lab_i(v)$ for~$v \in V(G_i)$ and~$i \in \{1,2\}$.
Thus the $t$-labeled composition first performs a $t$-labeling join 
of~$\bar{G}_1$ and~$g(\bar{G}_2)$ and then relabels the vertices
of~$G_1$ using~$f_1$ and the vertices of~$G_2$ with~$f_2$. Note that a
$t$-labeling composition is not commutative and that~$\{u,v\}$ is
an edge of~$\bar{H}$ if and only if~$\Lab_1(u) \bullet (\Lab_2(v) \times T_g) =1$, 
where~$T_g$ is the matrix representing the linear transformation~$g$.

\begin{definition}[$t$-labeled Parse Trees]
\rm
A \bem{$t$-labeled parse tree}~$T$ is a finite,
ordered, rooted subcubic tree (with the root of degree at most two) such that
\begin{enumerate}
\item all leaves of~$T$ are labeled with the~$\odot$ symbol, and
\item all internal nodes of~$T$ are labeled with a $t$-labeled composition symbol.
\end{enumerate}
A parse tree~$T$ \bem{generates} the graph~$G$ that is obtained by the successive 
leaves-to-root application of the operators that label the nodes of~$T$.
\end{definition}

The next result shows that rankwidth can be defined using $t$-labeled parse trees.

\begin{theorem}[The Rankwidth Parsing Theorem {\cite{ck07,GH10}}]
A graph~$G$ has rankwidth at most~$t$ if and only if some labeling of~$G$ can be
generated by a $t$-labeled parse tree. Moreover, a width-$k$ rank-decomposition
of an $n$-vertex graph can be transformed into a $t$-labeled parse tree on $\Theta(n)$
nodes in time~$O(t^2 \cdot n^2)$.
\end{theorem}
  
We now proceed to show the following. 
\def\themainthm{
\begin{main}[\cite{CMR00,GH10}]
Let~$\varphi$ be an \msoone-formula with~$\qr{\varphi} \leq q$. 
There is an algorithm that takes as input a $t$-labeled parse tree decomposition~$T$ of 
a graph~$G$ and decides whether~$G \models \varphi$ in time $O(f(q,t) \cdot |T|)$,
where~$f$ is some computable function and~$|T|$ is the number of nodes in~$T$. 
\end{main}
}
\themainthm
Here is how the sequel is organized. In Section~\ref{sec:msol} we briefly 
introduce monadic second order logic.
In Section~\ref{sec:q_equiv} we introduce 
a construct that plays a key role in our proof of the Main Theorem. 
This construct, called a characteristic tree of depth~$q$, 
is important for three reasons.
Firstly, a characteristic tree of depth~$q$ for a graph~$G$ allows
one to test whether an \mso\ formula~$\varphi$ of quantifier rank at 
most~$q$ holds in~$G$. Secondly, a characteristic tree has
small size and, thirdly, it can be efficiently constructed for
graphs of bounded rankwidth. 
The construction of characteristic trees is described in 
Section~\ref{sec:constructing_char_trees}, where we also
prove the main theorem.

\section{An Introduction to \mso\ Logic}\label{sec:msol}
In this section, we present a brief introduction to monadic second order
logic. We follow Ebbinghaus and Flum~\cite{EF99}.
\emph{Monadic second-order logic (MSOL)} is an extension of first-order logic
which allows quantification over sets of objects.
To define the
syntax of MSO, fix a \emph{vocabulary}~$\tau$ which is a finite set
of relation symbols~$P,Q,R, \ldots$ each associated with a natural
number known as its \emph{arity}.  

A \emph{structure~$\ms{A}$ over vocabulary~$\tau$} 
(also called a \emph{$\tau$-structure}) consists of a 
set~$A$ called the \emph{universe} of~$\ms{A}$ and a 
$p$-ary relation~$R^{\ms{A}} \subseteq A \times \cdots \times A$ ($p$ times)
for every $p$-ary relation symbol~$R$ in~$\tau$. If the universe 
is empty then we say that the structure is \emph{empty}. 
Graphs can be expressed in a natural way as relational
structures with universe the vertex set 
and a vocabulary consisting of a single binary (edge) relation symbol. To express
a $t$-labeled graph~$G$, we may use a vocabulary~$\tau$ consisting of
the binary relation symbol~$E$ (representing, as usual, the edge relation) 
and~$t$ unary relation symbols~$L_1, \ldots, L_t$, where~$L_i$ represents 
the set of vertices labeled~$i$.

A formula in MSO is a 
string of symbols from an alphabet that consists of
\begin{itemize}
\item the \emph{relation symbols} of~$\tau$
\item a countably infinite set of \emph{individual variables}~$x_1, x_2, \ldots$
\item a countably infinite set of \emph{set variables}~$X_1, X_2, \ldots$
\item $\neg$, $\vee$, $\wedge$ (the connectives \emph{not}, \emph{or}, \emph{and}) 
\item $\exists$, $\forall$ (the \emph{existential quantifier} and the
      \emph{universal quantifier})
\item $=$ (the \emph{equality} symbol)
\item $($, $)$ (the \emph{bracket} symbols).
\end{itemize}

The \emph{formulas} of MSO over the vocabulary~$\tau$ 
are strings that are obtained from finitely many applications
of the following rules:
\begin{enumerate}
\item If~$t_1$ and~$t_2$ are individual (respectively, set) 
variables then~$t_1 = t_2$ is a formula.
\item If~$R$ is an $p$-ary relation symbol in~$\tau$ 
and~$t_1, \ldots, t_r$ are individual variables, 
then~$Rt_1, \ldots, t_r$ is a formula.
\item If~$X$ is a set variable and~$t$ 
is an individual variable then~$Xt$ is a formula.
\item If~$\varphi$ is a formula then~$\neg \varphi$ is a formula.
\item If~$\varphi$ and~$\psi$ are formulas then~$(\varphi \vee \psi)$ is a
      formula.
\item  If~$\varphi$ and~$\psi$ are formulas then~$(\varphi \wedge \psi)$ is a
      formula.
\item If~$\varphi$ is a formula and~$x$ an individual variable
      then~$\exists x \varphi$ is a formula.
\item If~$\varphi$ is a formula and~$x$ an individual variable
      then~$\forall x \varphi$ is a formula.
\item If~$\varphi$ is a formula and~$X$ a set variable then~$\exists X \varphi$ is a formula.
\item If~$\varphi$ is a formula and~$X$ a set variable then~$\forall X \varphi$ is a formula.
\end{enumerate} 
The formulas obtained by~$1$, $2$, or~$3$ above are \emph{atomic formulas}.
Formulas of types~$6$, $8$, and $10$ are called~\emph{universal}, and formulas
of types~$5$, $7$, and $9$ are \emph{existential}.

The \emph{quantifier rank}~$\qr{\varphi}$ of a formula~$\varphi$ is the maximum
number of nested quantifiers occurring in it.
\[\begin{array}{rclrcl}
\qr{\varphi}           & := & 0, \, \, \mbox{if $\varphi$ is atomic}; \hspace{1cm}&
\qr{\exists x \varphi} & := & \qr{\varphi} + 1; \\
\qr{\neg \varphi}      & := & \qr{\varphi}; \hspace{1cm} &
\qr{\exists X \varphi} & := & \qr{\varphi} + 1; \\
\qr{\varphi \vee \psi} & := & \max\{\qr{\varphi}, \qr{\psi}\};
 \hspace{1cm}& 
\qr{\forall x \varphi} & := & \qr{\varphi} +1. \\
\qr{\forall X \varphi} & := & \qr{\varphi} +1; & & & \\ 
\end{array}\] 
A variable in a formula is \emph{free} if it is not within the scope of 
a quantifier. A formula without free variables is called a \emph{sentence}.
By $\free(\varphi)$ we denote the set of free variables of~$\varphi$.

We now assign meanings to the logical symbols by defining the
\emph{satisfaction relation}~$\ms{A} \models \varphi$. Let~$\ms{A}$ 
be a $\tau$-structure. An \emph{assignment} in $\ms{A}$ 
is a function~$\alpha$ that assigns individual variables
values in $A$ and set variables subsets of~$A$. 
For an individual variable~$x$ and an assignment~$\alpha$, we
let~$\alpha[x/a]$ denote an assignment that agrees
with $\alpha$ except that it assigns the value~$a \in A$ to~$x$. 
The symbol~$\alpha [X/B]$ has the same meaning for a set variable~$X$
and a set $B \subseteq A$.
We define
the relation~$\ms{A} \models \varphi[\alpha]$ ($\varphi$ is true in
$\ms{A}$ under $\alpha$) as follows:

\medskip

\begin{tabular}{lcl}
$\ms{A} \models t_1 = t_2[\alpha]$ & iff &
 $\alpha(t_1) = \alpha(t_2)$ \\
$\ms{A} \models Rt_1 \ldots t_n[\alpha]$ & iff & 
$R^{\ms{A}}\alpha(t_1) \ldots \alpha(t_n)$ \\
$\ms{A} \models \neg \varphi[\alpha]$ & iff & 
not $\ms{A} \models \varphi [\alpha]$ \\
$\ms{A} \models (\varphi \vee \psi)[\alpha]$ & iff & 
$\ms{A} \models \varphi [\alpha]$  or $\ms{A} \models \psi [\alpha]$ \\
$\ms{A} \models (\varphi \wedge \psi)[\alpha]$ & iff & 
$\ms{A} \models \varphi [\alpha]$  and $\ms{A} \models \psi [\alpha]$ \\
$\ms{A} \models \exists x \varphi [\alpha]$ & iff & 
there is an $a \in A$ such that $\ms{A} \models \varphi [\alpha[x/a]]$\\
$\ms{A} \models \forall x \varphi [\alpha]$ & iff & 
for all $a \in A$ it holds that $\ms{A} \models \varphi [\alpha[x/a]]$\\
$\ms{A} \models \exists X \varphi [\alpha]$ & iff & 
there exists $B \subseteq A$ such that $\ms{A} \models \varphi [\alpha[X/B]]$\\
$\ms{A} \models \forall X \varphi [\alpha]$ & iff & 
for all $B \subseteq A$ it holds that $\ms{A} \models \varphi [\alpha[X/B]]$\\
\end{tabular}

\section{The $\msoeq{q}$-Relation and its Characterization}
\label{sec:q_equiv}
Given a vocabulary~$\tau$ and a natural number~$q$, one can define an equivalence relation 
on the class of $\tau$-structures as follows. For $\tau$-structures $\ms{A}$ 
and~$\ms{B}$ and~$q \in \mathbf{N}$, define~$\ms{A} \msoeq{q} \ms{B}$ ($q$-equivalence) 
if and only if~$\ms{A} \models \varphi \Longleftrightarrow \ms{B} \models \varphi$  
for all \mso\ sentences~$\varphi$
of quantifier rank at most~$q$. In other words, two structures are $q$-equivalent if and only if
no sentence of quantifier rank at most~$q$ can distinguish them. 

We provide a characterization of the relation~$\msoeq{q}$ 
using objects called characteristic trees of depth~$q$. 
We show that two $\tau$-structures~$\ms{A}$ and~$\ms{B}$ have 
identical characteristic trees of depth~$q$ if and only 
if~$\ms{A} \msoeq{q} \ms{B}$. We shall see that characteristic trees  
are specially useful because their size is ``small'' and for graphs
of bounded rankwidth can be constructed efficiently given their 
parse tree decomposition. However before we can do that, we need 
a few definitions.

\begin{definition}[Induced Structure and Sequence]
\label{def:induced_str}
\rm
Let~$\ms{A}$ a $\tau$-structure with universe~$A$ 
and let~$\bar{c} = c_1, \ldots, c_m \in A^m$. The 
structure~$\ms{A}' = \ms A[\bar c] = \ms{A}[\{c_1, \ldots, c_m\}]$ induced 
by~$\bar c$ is a $\tau$-structure 
with universe~$A' = \{c_1, \ldots, c_m\}$ and 
interpretations~$P^{\ms{A}'} := P^{\ms{A}} \cap \{c_1, \ldots, c_m\}^r$
for every relation symbol~$P \in \tau$ of arity~$r$.

For an arbitrary sequence of objects $\bar c = c_1, \ldots, c_m$
and a set $U$, we let $\bar c[U]$ be the subsequence of $\bar c$ that contains only
objects in~$U$.  For a sequence of sets
$\bar{C} = C_1, \ldots, C_p$ we let~$\bar{C} \cap A$ 
denote the sequence $C_1 \cap U, \ldots, C_p \cap U$ and write $\bar C \cap \bar c$ for 
$C_1 \cap \{c_1, \ldots, c_m\}, \ldots, C_p \cap \{c_1, \ldots, c_m\}$.
\end{definition}

\begin{definition}[Partial Isomorphism]
\rm
Let~$\ms{A}$ and~$\ms{B}$ be structures over the vocabulary $\tau$ 
with universes~$A$ and~$B$, respectively, and let~$\pi$ be a map such 
that~$\dom{\pi} \subseteq A$ and~$\range{\pi} \subseteq B$.
The map~$\pi$ is said to be a \bem{partial isomorphism}
from~$\ms{A}$ to~$\ms{B}$
if
\begin{enumerate}
\item $\pi$ is one-to-one and onto;
\item for every $p$-ary relation symbol~$R \in \tau$ and 
all~$a_1, \ldots, a_p \in \dom{\pi}$,
\[R^{\ms{A}}a_1, \ldots,a_p \qquad \mbox{iff} \qquad
R^{\ms{B}}\pi(a_1), \ldots, \pi(a_p).\]
\end{enumerate}
If~$\dom{\pi} = A$ and $\range{\pi} = B$,
then~$\pi$ is an \bem{isomorphism} between~$\ms{A}$ and~$\ms{B}$
and~$\ms{A}$ and~$\ms{B}$ are \bem{isomorphic}.

Let $(\ms A, \bar A)$ and $(\ms B, \bar B)$ be tuples, where
$\bar A = A_1, \ldots, A_s$ and $\bar B = B_1, \ldots, B_s$, $s \ge 0$,
such that for all $1 \leq i \leq s$, we have $A_i \subseteq A$ and $B_i \subseteq B$. 
We say that $\pi$ is a partial isomorphism between
$(\ms A, \bar A)$ and $(\ms B, \bar B)$ if
\begin{enumerate}
\item $\pi$ is a partial isomorphism between $\ms A$ and $\ms B$,
\item for each $a \in \dom{\pi}$ and all $1 \le i \le s$, it holds that
$a \in A_i$ iff $\pi(a) \in B_i$.
\end{enumerate}
The tuples $(\ms A, \bar A)$ and $(\ms B, \bar B)$ are \bem{isomorphic}
if $\pi$ is an isomorphism between $\ms{A}$ and $\ms{B}$.
\end{definition}

In Definition~\ref{def:induced_str} of an induced structure we ignore
the order of the elements in~$\bar{c}$.  For the purposes in this paper,
the order in which the elements are chosen is important
because it is used to map variables in the formula to elements in the structure.
Moreover, elements could repeat in the vector~$\bar{c}$
and this fact is lost when we consider the induced structure~$\ms{A}[\bar{c}]$.
To capture both the order and the multiplicity of the elements in
vector~$\bar{c}$ in the structure~$\ms{A}[\bar{c}]$,
we introduce the notion of an \emph{ordered induced structure}.

Let $U$ be a set and $\equiv$ be an equivalence relation on $U$.  For $u
\in U$, we let $[u]_{\equiv} = \{\,u' \in U \mid u \equiv u'\,\}$ be
the \bem{equivalence class} of $u$ under~$\equiv$, and $U / {\equiv}
= \{\,[u]_{\equiv} \mid u \in U\}$ be the \bem{quotient space} of~$U$
under~$\equiv$.

A vector~$\bar{c} = c_1, \ldots, c_m \in A^m$ defines a natural
equivalence relation~$\eqv{\bar{c}}$ on the set~$[m] = \{1,\ldots,m\}$:
for~$i,j \in [m]$, we have~$i \eqv{\bar{c}}j$ if and only if~$c_i = c_j$.
For simplicity, we shall write $[i]_{\bar c}$ for $[i]_{\eqv{\bar c}}$.

\begin{definition}[Ordered Induced Structure]
\label{def:ordind}
\rm
Let $\ms{A}$ be a $\tau$-structure and~$\bar{c} = c_1, \ldots, c_m \in A^m$.
The \bem{ordered structure induced by~$\bar{c}$} is the $\tau$-structure~$\ms H = \ordind(\ms A, \bar c)$
with universe~$H = [m]/{\eqv{\bar c}}$ such that
the map $h\colon c_i \mapsto [i]_{\bar c}$, $1 \leq i \leq m$, is an
isomorphism between~$\ms{A}[\bar{c}]$ and~$\ms{H}$.

Let $\bar C = C_1, \ldots, C_p$ with $C_i \subseteq A$, $1 \le i \le p$.
Then we let
$$
\ordind(\ms A, \bar c, \bar C) := \left(\ordind(\ms A, \bar c), \bar h, h(\bar C \cap \bar c)\right),
$$
where $h\colon c_i \mapsto [i]_{\bar c}$, $1 \leq i \leq m$,
$\bar h = h(1),\ldots,h(m)$ and 
$h(\bar C \cap \bar c) = h(C_1 \cap \bar c),\ldots,h(C_p \cap \bar c)$.
\end{definition}

\begin{figure}[tb]
\centerline{\includegraphics{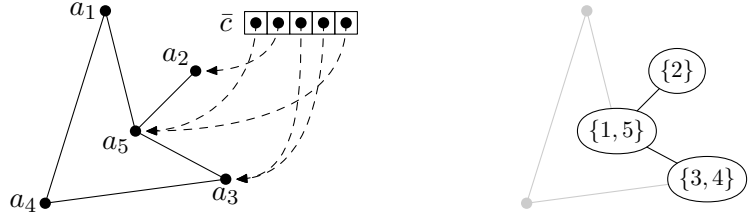}}
\caption{The vector~$\bar c = a_5a_2a_3a_3a_5$ lists vertices in the graph $\ms G$ on the left.
        The resulting ordered induced structure~$\ordind(\ms G, \bar c)$ is depicted in black on the right.
\label{fig:ordind}}
\end{figure}

Thus an ordered structure~$\ms H = \ordind(\ms A, \bar c)$ induced
by~$\bar{c}$ is simply the structure~$\ms{A}[\bar{c}]$
with element~$c_i$ being called~$[i]_{\bar c}$.
See Figure~\ref{fig:ordind} for an example.

\subsection{Model Checking Games and Characteristic Trees}
\label{subsec:model_check_char_trees}
Testing whether a non-empty structure models a formula can be specified 
by a \emph{model checking game} (also known as \emph{Hintikka game}, see
\cite{Hin73,Gra07}).
Let $\ms{A}$ be a $\tau$-structure with universe~$A$.  Let $\varphi$ be a
formula and $\alpha$ be an assignment to the free variables of~$\varphi$.
The game is played between
two players called the \emph{verifier} and the \emph{falsifier}. The
verifier tries to prove that $\ms{A} \models \varphi[\alpha]$ whereas the
falsifier tries to disprove this.
We assume without loss of generality that~$\varphi$ is in negation normal form, 
i.e., negations in~$\varphi$ appear only at the atomic level.  This can always
be achieved by applying simple rewriting rules such as
$\neg \forall x \varphi(x) \leadsto \exists x \neg \varphi(x)$.
The model checking game~$\mcg(\ms A, \varphi, \alpha)$
is positional with positions~$(\psi, \beta)$, where
$\psi$ is a subformula of~$\varphi$ and $\beta$ is an assignment to the free variables
of~$\psi$.  
The game starts at position~$(\varphi, \alpha)$.  At a
position $(\forall X \psi(X), \beta)$, the falsifier chooses a subset~$D \subseteq A$, 
and the game continues at position $(\psi, \beta[X/D])$.
Similarly, at a position $(\forall x \psi(x), \beta)$
or~$(\psi_1 \wedge \psi_2, \beta)$, the falsifier chooses
an element~$d \in A$ or some~$\psi := \psi_i$ for some~$1 \leq i \leq 2$ and
the game then continues at position~$(\psi, \beta[x/d])$ or $(\psi, \beta)$, respectively.
The verifier moves analogously at existential formulas.
If an element is chosen then the move is called a \emph{point move}; if a set is  
chosen then the move is a \emph{set move}.
The game ends once a position $(\psi, \beta)$ is reached, such that $\psi$ is an atomic
or negated formula.  The verifier \emph{wins} if and only if~$\ms{A} \models \psi[\beta]$.
We say that the verifier has a \emph{winning strategy} if they win every
play of the game irrespective of the choices made by the falsifier.

It is well known that the model checking game characterizes the satisfaction relation
$\models$.  The following lemma can easily be shown by induction over the structure
of~$\varphi$.

\begin{lemma}[cf., \cite{Gra07}]
\label{lemma:ef_model_checking}
Let~$\ms{A}$ be a $\tau$-structure, let~$\varphi$ be an \mso\ formula, and
let $\alpha$ be an assignment to the free variables of~$\varphi$.
Then~$\ms{A} \models \varphi[\alpha]$ if and only if the verifier has a winning strategy
on the model checking game on~$\ms{A}$, $\varphi$, and~$\alpha$.
\end{lemma}

A model checking game on a $\tau$-structure~$\ms{A}$ and 
a formula~$\varphi$ with quantifier rank~$q$ can be represented 
by a tree of depth~$q$ in which the nodes represent positions 
in the game and the edges represent point and set moves made by the players. 
Such a tree is called a \emph{game tree} and is used in combinatorial
game theory for analyzing games (see~\cite{BCG82}, for 
instance).

For our purposes, we define a notion related to game trees 
called \emph{full characteristic trees} which are finite rooted trees,
where the nodes represent positions and edges
represent moves of the game. A node is a tuple that represents 
the sets and elements that have been chosen thus far.
The node can be thought of as a succinct representation of the
state of the game played till the position represented by that node.
However, note that a full characteristic tree depends on the
quantifier rank~$q$ and \emph{not} on a particular formula.

\begin{definition}[Full Characteristic Trees]
\label{defn:full_char_tree}
\rm
Let~$\ms{A}$ be a $\tau$-structure with universe~$A$ 
and let~$q \in \mathbf{N}$. For elements
$\bar{c} = c_1, \ldots, c_m \in A^m$, sets~$\bar{C} = C_1, \ldots, C_p$ 
with~$C_i \subseteq A$, $1 \le i \le p$, let~$T = \fullchar{q}{\ms{A}}{\bar{c}}{\bar{C}}$
be a finite rooted tree such that
\begin{enumerate}
\item $\root(T) = (\ms{A}[\bar{c}],\bar{c},\bar{C} \cap \bar{c})$,
\item if~$m+p+1 \leq q$ then the subtrees of the root of 
$\fullchar{q}{\ms{A}}{\bar{c}}{\bar{C}}$
is the set
$$
%\subtrees(\fullchar{q}{\ms{A}}{\bar{c}}{\bar{C}}) = 
	\big\{\, \fullchar{q}{\ms{A}}{\bar{c} d}{\bar{C}} \bigm| d \in A\,\big\} \cup  
	\big\{\, \fullchar{q}{\ms{A}}{\bar{c}}{\bar{C} D} \bigm| D \subseteq A\,\big\}.
$$
\end{enumerate}
The \bem{full characteristic tree of depth~$q$} for~$\ms{A}$, 
denoted by $\fc{q}{\ms{A}}$, is defined 
as~$\fullchar{q}{\ms{A}}{\varepsilon}{\varepsilon}$, 
where~$\varepsilon$ is the empty sequence. 
\end{definition}

Let $T = (V, E)$ be a rooted tree.  We let $\root(T)$ be the root of
$T$ and for~$u \in V$ we let~$\children_T(u) = \{\,v \in V \mid (u,v) \in E \,\}$
and $\subtree_T(u)$ be a subtree of~$T$ rooted at~$u$,
and~$\subtrees(T) = \{\,\subtree_T(u) \mid u \in \children_T(\root(T))\,\}$. 

We now define a model checking game $\mcg(F, \varphi, \bar x, \bar X)$
on full characteristic trees~$F = \fullchar{q}{\ms{A}}{\bar c}{\bar C}$ and formulas~$\varphi$ with~$\qr{\varphi} \leq q$,
where $\bar x = x_1, \ldots, x_m$ are the free object variables of $\varphi$,
$\bar X = X_1, \ldots, X_p$ are the free set variables of $\varphi$,
$\bar c = c_1, \ldots, c_m \in A^m$, and $\bar C = C_1, \ldots, C_p$ with $C_i \subseteq A$, $1 \le i \le p$.
The rules are similar to the classical model checking game $\mcg(\ms{A}, \varphi, \alpha)$.
The game is positional and played
by two players called the \emph{verifier} and the \emph{falsifier} and
is defined over subformulas~$\psi$ of~$\varphi$.
However instead of choosing sets and elements explicitly, the tree~$F$
is traversed top-down.  At the same time, we ``collect'' the list of variables the players encountered,
such that we can make the assignment explicit once the game ends.
The game starts at the position~$(\varphi, \bar x, \bar X, \root(F))$.
Let $(\psi, \bar y, \bar Y, v)$ be the position at which the game is being played, where $v = (\ms H, \bar d, \bar D)$
is a node of $\fullchar{q}{\ms A}{\bar c}{\bar C}$, and $\psi$ is a subformula of~$\varphi$ with
$\free(\psi) = \bar y \cup \bar Y$.
At a position $(\forall X \vartheta(X), \bar y, \bar Y, v)$ the falsifier chooses a child $u =
(\ms H, \bar d, \bar D D)$ of $v$, where $D \subseteq A$, and the game
continues at position~$(\vartheta, \bar y, \bar Y X, u)$.  Similarly, at a position $(\forall
x \vartheta(x), \bar y, \bar Y, v)$ the falsifier chooses a child $u = (\ms H', \bar d d,
\bar D)$, where $d \in A$, and the game continues in $(\vartheta, \bar y x, \bar Y, u)$, and at a position
$(\vartheta_1 \wedge \vartheta_2, \bar y, \bar Y, v)$, the falsifier chooses some $1 \le i \le 2$,
and the game continues at position~$(\vartheta_i, \bar y, \bar Y, v)$.
The verifier moves analogously at existential formulas.

The game stops once an atomic or negated 
formula has been reached. 
Suppose that a particular play of the game 
ends at a position $(\psi, \bar y, \bar Y, v)$, where 
$\psi$ is a negated atomic or atomic formula with 
\[\free(\psi) = \{y_1, \ldots, y_{s}, Y_1, \ldots, Y_t\}\] 
and $v = (\ms{H}, \bar{d}, \bar{D})$ some node of~$F$, where
$\bar d = d_1, \ldots, d_s$ and $\bar D = D_1, \ldots, D_t$.
Let $\alpha$ be an assignment to the free variables of $\varphi$, such that
$\alpha(y_i) = d_i$, $1 \le i \le s$, and $\alpha(Y_i) = D_i$, $1 \le i \le t$.
The verifier \emph{wins} the game if and only if $\ms{H} \models \psi[\alpha].$
The verifier has a \emph{winning strategy} if and only if they can win every play of the game
irrespective of the choices made by the falsifier.
In what follows, we identify a position $(\psi, \bar y, \bar Y, v)$
of the game $\mcg(\fullchar{q}{\ms A}{\bar c}{\bar C}, \varphi, \bar x, \bar X)$, where
$v = (\ms H, \bar d, \bar D)$, with the game $\mcg(\fullchar{q}{\ms A}{\bar d}{\bar D}, \psi, \bar y, \bar Y)$.

\begin{lemma}\label{lemma:games_full_char_trees}
Let~$\ms{A}$ be a $\tau$-structure and let~$\varphi$ be an \mso\ formula
with $\qr{\varphi} \leq q$ and free variables $\{x_1, \ldots,x_m, X_1, \ldots, X_m\}$. 
Let $\alpha$ be an assignment to the free variables of $\varphi$. 
Then the verifier has a winning strategy in the model checking game $\mcg(\ms A, \varphi, \alpha)$ if
and only if the verifier has a winning strategy in the model checking game $\mcg(\fullchar{q}{\ms A}{\bar c}{\bar C}, \varphi, \bar x, \bar X)$,
where $\bar c = \alpha(x_1), \ldots, \alpha(x_m)$ and $\bar C = \alpha(X_1), \ldots, \alpha(X_p)$.
\end{lemma}
\begin{proof}
The proof consists in observing that any play of the model checking 
game~$\mcg(\ms A, \varphi, \alpha)$ can be simulated in 
$\mcg(\fullchar{q}{\ms A}{\bar c}{\bar C}, \varphi, \bar x, \bar X)$ and vice versa.
\qed
\end{proof}

Lemma~\ref{lemma:games_full_char_trees} showed that a full characteristic
tree of depth~$q$ for a structure~$\ms{A}$ can be used to simulate the model checking
game on~$\ms{A}$ and any formula~$\varphi$ of quantifier rank at most~$q$.
However the size of such a tree is of the order~$(2^n + n)^q$, where~$n$
is the number of elements in the universe of~$\ms{A}$. We now show that one
can ``collapse'' equivalent branches of a full characteristic tree to
obtain a much smaller labeled tree (called a reduced characteristic tree) 
that is in some sense equivalent to the original (full) tree. We will then show
that for a graph~$G$ of rankwidth at most~$t$, the reduced characteristic tree of~$G$ 
is efficiently computable given a $t$-labeled parse tree decomposition 
of~$G$. We achieve this collapse by replacing the induced structures 
$\ms A[\bar c]$ in the full characteristic tree by a more generic, 
implicit representation --- that of their ordered induced 
substructures $\ordind(\ms A, \bar c)$. 

\begin{definition}[Reduced Characteristic Trees]
\rm
Let~$\ms{A}$ be a $\tau$-structure and let $q \in \mathbf{N}$.
For elements $\bar{c} = c_1, \ldots, c_m \in A^m$ and sets $\bar{C} = C_1, \ldots, C_p$ 
with~$C_i \subseteq A$, $1 \le i \le p$, we let $\redchar{q}{\ms{A}}{\bar{c}}{\bar{C}}$
be a finite rooted tree such that
\begin{enumerate}
\item $\root(\redchar{q}{\ms{A}}{\bar{c}}{\bar{C}}) = \ordind(\ms A, \bar c, \bar C)$,
\item if~$m+p+1 \leq q$ then the subtrees of the root of 
$\redchar{q}{\ms{A}}{\bar{c}}{\bar{C}}$
is the set
$$
%\subtrees(\redchar{q}{\ms{A}}{\bar{c}}{\bar{C}}) = 
	\{\, \redchar{q}{\ms{A}}{\bar{c} d}{\bar{C}} \mid d \in A\,\} \cup 
	\{\, \redchar{q}{\ms{A}}{\bar{c}}{\bar{C} D} \mid D \subseteq A\,\}.
$$
\end{enumerate}
The \bem{reduced characteristic tree of depth~$q$} 
for the structure~$\ms{A}$, denoted by $\rc{q}{\ms{A}}$,  
is defined to be $\redchar{q}{\ms{A}}{\varepsilon}{\varepsilon}$,
where~$\varepsilon$ is the empty sequence. 
\end{definition}

\begin{figure}[tbp]
\centerline{\includegraphics{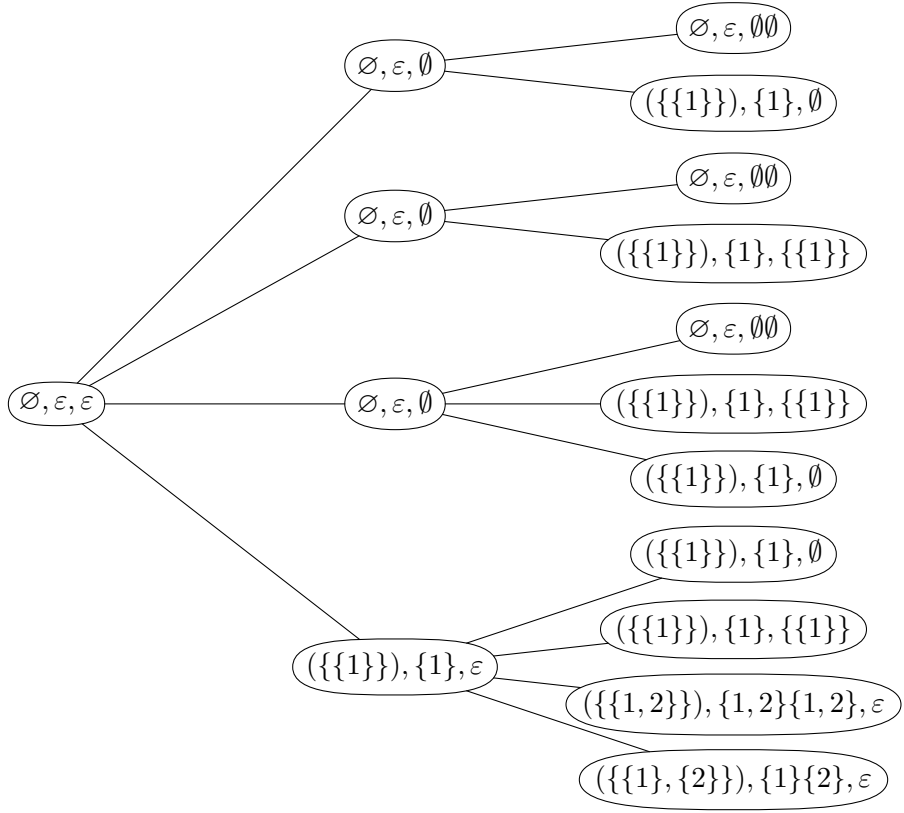}}
\caption{The tree $\rc{2}{\ms A}$ for a $\tau$-structure $\ms A$ with $\tau = \emptyset$ and
$A = \{a_1,a_2\}$.  Here, $\varnothing$ denotes an empty structure, and $\emptyset\emptyset$ is the sequence of 
two empty sets.  The bottom right node $(\ms H, \bar c, \bar C) = \big((\{\{1\},\{2\}\}), \{1\}\{2\}, \varepsilon\big)$
represents, at the same time, the identical subtrees
$\redchar{2}{\ms A}{a_1a_2}{\varepsilon}$ and $\redchar{2}{\ms A}{a_2a_1}{\varepsilon}$.
The universe of $\ms H$ is $H = \{[1]_{a_1a_2}, [2]_{a_1a_2}\} = \{[1]_{a_2a_1}, [2]_{a_2a_1}\} = \{\{1\},\{2\}\}$, since elements $a_1, a_2$
and $a_2, a_1$, respectively, have been chosen in this order. 
No set has been chosen, hence the empty sequence $\bar C = \varepsilon$.
Similarly, the next node in that column, $\big((\{\{1,2\}\}), \{1,2\}\{1,2\}, \varepsilon\big)$,
represents the trees $\redchar{2}{\ms A}{a_1a_1}{\varepsilon}$
and $\redchar{2}{\ms A}{a_2a_2}{\varepsilon}$.  Here the universe is $\{\{1,2\}\}$ since the same element
has been chosen twice.  Note that the root node has only five subtrees in total, since
$\redchar{2}{\ms A}{\varepsilon}{\{a_1\}} = \redchar{2}{\ms A}{\varepsilon}{\{a_2\}}$ (third 
subtree from the top), and $\redchar{2}{\ms A}{a_1}{\varepsilon} = \redchar{2}{\ms A}{a_2}{\varepsilon}$
(bottom subtree).
\label{fig:redchar}
}
\end{figure}

See Figure~\ref{fig:redchar} for an example.
One can define the model checking game $\mcg(R, \varphi, \bar x, \bar X)$ on a 
tree $R = \redchar{q}{\ms{A}}{\bar c}{\bar C}$
in exactly the same manner as $\mcg(\fullchar{q}{\ms{A}}{\bar c}{\bar C}, \varphi, \bar x, \bar X)$.
As mentioned before, our interest in $\redchar{q}{\ms{A}}{\bar c}{\bar C}$ lies in that:
\begin{enumerate}
\item they are equivalent to $\fullchar{q}{\ms{A}}{\bar c}{\bar C}$,
\item they are ``small''; and,
\item they are efficiently computable if~$\ms{A}$ is a graph of rankwidth at most~$t$. 
\end{enumerate}

We first show that the reduced characteristic tree $\redchar{q}{\ms A}{\bar{c}}{\bar{C}}$ 
is equivalent to its full counterpart $\fullchar{q}{\ms{A}}{\bar c}{\bar C}$.
\begin{lemma}
\label{lemma:games_red_char_trees}
Let~$\ms{A}$ be a $\tau$-structure and let $q \in \mathbf{N}$.
Let $\bar{c} = c_1, \ldots, c_m \in A^m$ and $\bar{C} = C_1, \ldots, C_p$ with~$C_i \subseteq A$, $1 \le i \le p$.
Let $F = \fullchar{q}{\ms A}{\bar c}{\bar C}$ and 
$R = \redchar{q}{\ms A}{\bar c}{\bar C}$.
Then the verifier has a winning strategy in the model checking game
$\mcg(F, \varphi, \bar x, \bar X)$
if and only if the verifier has a winning strategy in the game
$\mcg(R, \varphi, \bar x, \bar X)$, where
$\varphi \in \mso(\tau)$ with $\qr{\varphi} \le q$ with free object
variables $\bar x = x_1, \ldots, x_m$ and free set variables $\bar X = X_1, \ldots, X_p$.
\end{lemma}
\begin{proof}
Without loss of generality, we assume $\qr{\varphi} = q$ (otherwise, pad~$\varphi$ with quantifiers).
The proof is by 
an induction on $q - m - p$ and the structure of~$\varphi$.
If $q = 0$, then
\begin{multline*}
\root(F) = (\ms A[\bar c], \bar c, \bar C \cap \bar c) \cong \\
(\ms H, h(1)\ldots h(m), h(\bar C \cap \bar c)) = \ordind(\ms A, \bar c, \bar C) = \root(R),
\end{multline*}
where $h\colon c_i \mapsto [1]_{\bar c}$, $1 \le i \le m$ is an isomorphism
between $(\ms A[\bar c], \bar C \cap \bar c)$ and $(\ms H, h(\bar C \cap \bar c))$.
The lemma therefore holds since MSO formulas cannot distinguish isomorphic structures.

Therefore assume that $q > 0$.
If $\varphi = (\psi_1 \wedge \psi_2)i$ or $\varphi = (\psi_1 \vee \psi_2)$, then the claim
immediately follows by the induction hypothesis for $\psi_i$, $1 \le i \le 2$.
Assume therefore that $\varphi = \exists X \psi(X)$ and suppose that the verifier has a winning
strategy in one of the games, say, in $\mcg(R, \varphi, \bar x, \bar X)$.  Then there is
a position $(\psi, \bar x, \bar X X, u)$, where $u \in \children_{R}(\root(R))$,
such that the verifier has a winning strategy in $\mcg(\subtree_{R}(u), \psi, \bar x, \bar X X)$
where $\subtree_{R}(u) = \redchar{q}{\ms A}{\bar c}{\bar C D}$ for some $D \subseteq A$.
By the induction hypothesis, the verifier has a winning strategy in $\mcg(F', \psi, \bar x, \bar X X)$,
where $F' = \fullchar{q}{\ms A}{\bar c}{\bar C D} \in \subtrees(F)$.
The verifier can therefore win $\mcg(F, \varphi, \bar x, \bar X)$ by choosing a
position~$(\psi, \bar x, \bar X X, \root(F'))$, which implies the claim.

If $\varphi = \forall x \psi(x)$, and the verifier has a winning
strategy in strategy in one of the games, say in $\mcg(R, \varphi,
\bar x, \bar X)$, consider a move of the falsifier to a position
$(\psi, \bar x x, \bar X, u)$ in $\mcg(F, \varphi, \bar x, \bar X)$,
where $u = \root(\fullchar{q}{\ms A}{\bar c d}{\bar C}))$ for some
$d \in A$.  Let $R' = \redchar{q}{\ms A}{\bar c d}{\bar C}$ be a subtree 
of the root of~$R$. %\in \subtrees(R).
The verifier has a winning strategy in the game $\mcg(R', \psi, \bar x x, \bar X)$,
and therefore, by the induction hypothesis, in $\mcg(\fullchar{q}{\ms
A}{\bar c d}{\bar C}, \psi, \bar x x, \bar X)$.

The remaining cases follow analogously.
\qed
\end{proof}

From Lemmas~\ref{lemma:ef_model_checking},
\ref{lemma:games_full_char_trees}, and~\ref{lemma:games_red_char_trees}, 
we obtain the important fact that reduced characteristic trees 
are in fact equivalent to their full counterparts and 
characterize the equivalence relation~$\msoeq{q}$.
\begin{corollary}
Let~$\ms{A}$ and~$\ms{B}$ be $\tau$-structures 
and~$q \in \mathbf{N}$. Then $\rc{q}{\ms{A}} = \rc{q}{\ms{B}}$ iff 
$\ms{A} \msoeq{q} \ms{B}$.
\end{corollary}

The next lemma shows that reduced characteristic trees 
have small size. For~$i \in \mathbf{N}$,  we
define~$\exp^{(i)}(\cdot)$ as: $\exp^{(0)}(x) = x$, $\exp^{(1)}(x) = 2^x$ 
and~$\exp^{(i)}(x) = 2^{2 \exp^{(i-1)}(x)}$ for~$i \geq 2$.
\begin{lemma}\label{lemma:red_char_size}
Let~$\ms{A}$ be a $\tau$-structure with universe~$A$
such that each relation symbol in~$\tau$ has arity at most~$r$, 
and~$q \in \mathbf{N}$.
Then the number of reduced characteristic trees~$\redchar{q}{\ms{A}}{\bar{c}}{\bar{C}}$ 
for all possible choices of~$\bar{c},\bar{C}$ 
is at most $\exp^{(q+1)}(|\tau| \cdot q^r + q \log q + q^2)$. The size of 
a reduced characteristic tree~$\redchar{q}{\ms{A}}{\bar{c}}{\bar{C}}$ 
is at most $(\exp^{(q)}(|\tau| \cdot q^r + q \log q + q^2))^4$. 
\end{lemma}

\begin{proof}
For integers~$m,p$ let $N(\ms{A},m,p)$ be the number of
trees $\redchar{q}{\ms{A}}{\bar{c}}{\bar{C}}$, where 
$\bar{c} = c_1, \ldots, c_m \in A^m$ and $\bar{C} = C_1, \ldots, C_p$ with
$C_i \subseteq A$, $1 \le i \le p$.
Define
\[S(\ms{A},m,p) = \max_{\bar{c},\bar{C}} |\redchar{q}{\ms{A}}{\bar{c}}{\bar{C}}|,\]
where the maximum is taken over all strings~$\bar{c}$ and~$\bar{C}$ such 
that~$|\bar{c}| = m$ and~$|\bar{C}| =p$. Also define
$f(\tau,q) = |\tau| \cdot q^r + q \log q + q^2$.

If~$m+p = q$ then $\redchar{q}{\ms{A}}{\bar{c}}{\bar{C}}$
has one node for all~$\bar{c}, \bar{C}$ 
and $S(\ms{A},m,p) =1$. 
The number of distinct trees $N(\ms{A}, m,p)$, however, depends on
the number of structures on a universe of size at most~$m \leq q$
over a vocabulary with~$|\tau|$ relation symbols 
each of arity at most~$r$.  The number of such structures is at most~$2^{|\tau| \cdot q^r}$,
and since there are at most~$q^q \cdot 2^{q^2}$ vectors~$\bar{c},\bar{C}$ over 
the $m + p \le q$ elements, we have that~$N(\ms{A},m,p) \leq 2^{f(\tau,q)} 
\le \exp^{(1)}(f(\tau,q))$.
If~$m+p < q$ then the root of $\redchar{q}{\ms{A}}{\bar{c}}{\bar{C}}$ can have
as children any of the $N(\ms{A}, m+1,p)$ reduced characteristic trees 
corresponding to point moves and $N(\ms{A},m,p+1)$ trees corresponding to set moves. 
Hence $N(\ms{A},m,p) \leq 2^{N(m+1,p) + N(m,p+1)}$. By induction hypothesis, each of $N(\ms{A},m+1,p)$
and $N(\ms{A}, m,p+1)$ is at most $\exp^{(q - (m+p))}(f(\tau,q))$
and hence 
\[
N(\ms{A},m,p) \leq 2^{2 \cdot \exp^{(q - (m+p))}(f(\tau,q))} = 
\exp^{(q - (m+p) + 1)}(f(\tau,q)).
\]
Hence $N(\ms{A},0,0) \leq \exp^{(q+1)}(f(\tau,q))$ as claimed. 

The size of a reduced characteristic tree is one if~$m+p = q$. Otherwise
\begin{multline*}
S(\ms{A},m,p) \leq 1 + S(\ms{A},m+1,p)N(\ms{A},m+1,p) + {} \\
              S(\ms{A},m,p+1)N(\ms{A},m,p+1),
\end{multline*}
since any such tree consists of a single root vertex and at most $N(\ms{A},m+1,p)$ 
trees (corresponding to point moves) each of size $S(\ms{A},m+1,p)$ and at most 
$N(\ms{A},m,p+1)$ trees (corresponding to set moves) of size $N(\ms{A},m,p+1)$. 
By induction hypothesis, each of the terms $S(\ms{A},m+1,p)$ and $S(\ms{A},m,p+1)$
is at most $(\exp^{(q - (m+p+1))}(f(\tau,q)))^4$ and hence
\[
S(\ms{A},m,p) \leq  1 + 2 \exp^{(q-(m+p))}(f(\tau,q)) \cdot (\exp^{(q - (m+p+1))}(f(\tau,q)))^4.
\]
One can show that the right hand side of the above inequality 
is at most $(\exp^{(q - (m+p))}(f(\tau,q)))^4$, 
thereby proving the claimed size bound. 
\qed
\end{proof}

\section{Constructing Characteristic Trees}
\label{sec:constructing_char_trees}

In this section, we show how to construct reduced characteristic trees of depth~$q$
for a graph~$G$ of rankwidth~$t$ when given a $t$-labeled parse tree decomposition 
of~$G$. A $t$-labeled graph may be represented as $\tau$-structure 
where~$\tau = \{E, L_1, \ldots, L_t\}$. The symbol~$E$ is a  binary relation symbol
representing the edge relation and~$L_i$ for~$1 \le i \le t$ is a unary relation symbol
representing the set of vertices with label~$i$. In what follows, whenever we talk about 
a $\tau$-structure $\ms{A}$, we mean a graph viewed as a structure over the vocabulary 
$\{E, L_1, \ldots, L_t\}$.

\begin{lemma}\label{lemma:construction_constant_size}
Let $\ms A$ be a $\tau$-structure with $|A| = 1$. Let $q \ge 0$ and $\bar c \in A^m$ and $\bar C = C_1,\ldots,C_p$ with
$C_i \subseteq A$, $1 \le i \le p$.
Then $\redchar{q}{\ms A}{\bar c}{\bar C}$ can be constructed in constant time for each fixed~$q$.
\end{lemma}
\begin{proof}
Note that, in this case, $\fullchar{q}{\ms A}{\bar c}{\bar C}$ has size at most~$O((2^1 + 1)^q) = O(3^q)$. 
Hence for each fixed~$q$, $\redchar{q}{\ms A}{\bar c}{\bar C}$ can be constructed in constant time.
\qed
\end{proof}

In what follows, we let $\ms A_1, \ms A_2$ and $\ms A = \ms A_1 \otimes \ms A_2$ be
$\tau$-structures, where $\otimes = {\myotimes{g}{f_1}{f_2}}$ for $t$-relabelings~$g$, $f_1$,
and~$f_2$. Recall that if $\ms A = \ms A_1 \otimes \ms A_2$, then we assume that 
$A_1$ and~$A_2$ (the universes of $\ms{A}_1$ and~$\ms{A}_2$, respectively) are disjoint. 
Furthermore for a fixed constant~$q \ge 0$, let~$m$ and~$p$ be nonnegative integers 
such that~$m + p \le q$, $\bar c = c_1,\ldots,c_m \in (A_1 \cup A_2)^m$
and $\bar C = C_1, \ldots, C_p$, where $C_j \subseteq A_1 \cup A_2$, $1 \le j \le p$.
For $i \in \{1,2\}$, we let $\bar c_i = c_{i,1},\ldots,c_{i,m_i} = \bar c[A_i]$.

In the remainder of this section, we show how to construct
$\redchar{q}{\ms A}{\bar c}{\bar C}$ given
$\redchar{q}{\ms A_1}{\bar c_1}{\bar C \cap \bar c_1}$ and
$\redchar{q}{\ms A_2}{\bar c_2}{\bar C \cap \bar c_2}$.
For the construction, we need to know the order in which 
the elements in $\bar c_1$ and $\bar c_2$
appear in $\bar c$.  This motivates us to define
the notion of an \emph{indicator vector} $\indicator(A_1, A_2, \bar c)$.

\begin{definition}
\rm
The \bem{indicator vector} of~$\bar c = c_1, \ldots, c_m$, denoted 
$\indicator(A_1, A_2, \bar c)$, is the vector $\bar d = d_1,\ldots,d_m$, such that
for $i \in \{1,2\}$ and all $1 \le j \le m$ it holds that $d_j = (i,k)$ iff  $c_j = c_{i,k}$.
That is, $d_j = (i, k)$ iff $c_j$ is the $k$th element in the vector~$\bar c_i = \bar c[A_i]$.
If $\bar d = d_1,\ldots,d_m$ and~$(i,k) \in \{1,2\} \times [m+1]$, then we use~$\bar{d}(i,k)$
to denote the vector~$d_1, \ldots, d_{m+1}$, where~$d_{m+1} = (i,k)$. 
\end{definition}

\begin{example}
\label{ex:indicator}
Let $A_1 = \{a_1,a_2\}$, $A_2 = \{b_1,b_2,b_3,b_4\}$ and let $\bar{c}$ be the 
string $a_1 b_1 b_2 a_2 b_3 b_4 a_2 b_3 a_1$.  Then we get:
{\small
$$
\begin{array}{cccccccccccc}
\bar c 			\: & = & \:   a_1 \: & b_1 \: & b_2 \: & a_2 \: & b_3 \: & b_4 \: & a_2 \: & b_3 \: & a_1 \\
\bar c[A_1] 		\: & = & \:   a_1 \: &     \: &     \: & a_2 \: &     \: &    \: & a_2 \: &     \: & a_1 \\
\bar c[A_2] 		\: & = & \:       \: & b_1 \: & b_2  \:&     \: & b_3 \: & b_4 \: &     \: & b_3 \: &     \\
\indicator(A_1,A_2,\bar c) \: & = & \: (1,1) \: & (2,1) \: & (2,2) \: & (1,2) \: & (2,3) \: & (2,4) \: & (1,3) \: & 
                                       (2,5) \: & (1,4)
\end{array}
$$
}
Given $\bar c[A_1]$, $\bar c[A_2]$, and $\bar d = d_1,\ldots,d_m = \indicator(A_1,A_2,\bar c)$, one can now reconstruct
$\bar c$.  For example, $c_8 = b_3$, since $d_8 = (2,5)$, which tells us that $c_8$ is the fifth element in $\bar c_2$.
\end{example}

Constructing $R = \redchar{q}{\ms A}{\bar c}{\bar C}$ when given
$R_1 = \redchar{q}{\ms A_1}{\bar c_1}{\bar C \cap \bar c_1}$,
$R_2 = \redchar{q}{\ms A_2}{\bar c_2}{\bar C \cap \bar c_2}$, and $\bar d = \indicator(A_1, A_2, \bar c)$ 
consists of the following two steps:
\begin{enumerate}
\item construct the label for $\root(R) = \ordind(\ms A, \bar c, \bar C)$, and then
\item recursively construct its subtrees.
\end{enumerate}
Since $\ordind(\ms A, \bar c) \cong \ms A[\bar c]$ and $\ms A_i[\bar c_i]
\cong \ordind(\ms A_i, \bar c_i)$, one easily sees that $$\ordind(\ms
A, \bar c) \cong \ordind(\ms A_1, \bar c_1) \otimes \ordind(\ms A_2,
\bar c_2).$$  For the first step, we therefore just need to rename elements in $\ordind(\ms A_1,
\bar c_1) \otimes \ordind(\ms A_2, \bar c_2)$ in an appropriate way.
The information on how elements are to be renamed is stored in
the indicator vector~$\bar d$ of $\bar c$.  See Figure~\ref{fig:moreotimes}
for an example.
The formal definition of the renaming operator $\moreotimes{\bar d}$
and Lemma~\ref{lem:moreotimes} are technical
and may be skipped if the reader believes that one can construct $\ordind(\ms A, \bar c)$ 
from $\ordind(\ms A_1, \bar c_1)$ and $\ordind(\ms A_2,\bar c_2)$ using~$\bar d$.

\begin{definition}
\label{def:moreotimes}
\rm
For $i \in \{1,2\}$, let $\ordind(A_i, \bar c_i, \bar C \cap A_i) = (\ms H_i, \bar c_i', \bar C_i')$.
Define a map $f\colon [m] \to H_1 \uplus H_2$ as follows: for all $1 \le j \le m$, 
let $f(j) = [k]_{\bar c_i}$ iff $d_j = (i, k)$.
Then we define
$\ordind(\ms A_1, \bar c[A_1], \bar C \cap A_1)
\moreotimes{\bar d} 
\ordind(\ms A_2, \bar c[A_2], \bar C \cap A_2)
$
as
$$\ordind(\ms H_1 \otimes \ms H_2,  f(1)\ldots f(m), \bar C_1' \cup \bar C_2').$$
\end{definition}

\begin{figure}[tb]
\centerline{\includegraphics{rankwidth.4}}
\caption{$\ms G_1$ and $\ms G_1$ depicted on the top left are graphs such that $\ms G_1 \oplus \ms G_2$ is the graph of Figure~\ref{fig:ordind};
the grey edges being those created by the $t$-labeled composition operator~$\oplus$.
For $\bar c = a_5a_2a_3a_3a_5$ and $\bar c_1 = \bar c[G_1]$, $\bar c_2 = \bar c[G_2]$ 
the ordered induced substructures $\ms H_1 = \ordind(\ms G_1, \bar c_1)$ and $\ms H_2 = \ordind(\ms G_1, \bar c_2)$ depicted in black on the bottom left.
On these, we can take the $t$-labeled composition $\ms H = \ms H_1 \oplus \ms H_2$ and obtain the graph isomorphic to $\ms G_1[\bar c_1] \oplus \ms G_2[\bar c_2]$
on the bottom right.  We can now use the vector~$\bar d = (1,1)(2,1)(2,2)(2,3)(1,2)$ to rename vertices
in $\ms H$ and obtain $\ordind(\ms G, \bar c)$ depicted on the top right.
Note that $\bar c$ and $\bar d$ essentially describe the same vertices.
\label{fig:moreotimes}}
\end{figure}

\begin{lemma}
\label{lem:moreotimes}
Let~$\ms A_1$ and $\ms A_2$ be $\tau$-structures and 
let $\otimes = \myotimes{g}{f_1}{f_2}$ for 
 some $t$-relabelings $g, f_1, f_2$. Let $\bar c = c_1,\ldots,c_m \in (A_1 \cup A_2)^m$
and $\bar C = C_1, \ldots, C_p$, where $C_j \subseteq A_1 \cup A_2$ for 
$1 \le j \le p$. Also let $\bar{d} = \indicator(A_1,A_2,\bar{c})$. Then
%Using notation from Definition~\ref{def:moreotimes}, we have
$$
\ordind(\ms A_1 \otimes \ms A_2, \bar c, \bar C) =
	\ordind(\ms A_1, \bar c[A_1], \bar C \cap A_1)
	\moreotimes{\bar d} 
	\ordind(\ms A_2, \bar c[A_2], \bar C \cap A_2).
$$
\end{lemma}
\begin{proof}
For $i \in \{1,2\}$, it holds
$$
\ordind(\ms A_i, \bar c_i, \bar C_i) = (\ms H_i, \bar c_i', \bar C_i') \cong (\ms A_i[\bar c[A_i]], \bar c[A_i], \bar C \cap A_i),
$$
where $h_i \colon c_{i,j} \mapsto [j]_{\bar c_i}$, $1 \le j \le m_i$ is the isomorphism
of Definition~\ref{def:ordind} and
$\bar c_i' = c_{i,1}', \ldots, c_{i,m_i}' = h_i(1),\ldots,h_i(m_i) \in H_i^{m_i}$.
Let $\ms H = \ms H_1 \otimes \ms H_2$ be the $\tau$-structure with universe
$H = H_1 \uplus H_2 = [m_1]/{\eqv{\bar c_1}} \uplus [m_2]/{\eqv{\bar c_2}}$, where we assume
without loss of generality that $H_1$ and $H_2$ are disjoint (rename elements otherwise).
We want to show the equality in the following diagram (see also Figure~\ref{fig:moreotimes} for a concrete example):
$$
\begin{array}{ccccccccccc}
\ms A_1[\bar c_1] & \otimes & \ms A_2[\bar c_2] & = & \ms A[\bar c] & \qquad \cong \qquad & \ordind(\ms A, c_1\ldots c_m) \\ 
\reflectbox{\rotatebox[origin=c]{90}{$\cong$}}
&&
\reflectbox{\rotatebox[origin=c]{90}{$\cong$}}
&&&&
\reflectbox{\rotatebox[origin=c]{90}{$=$}}
\\
\ms H_1 & \otimes & \ms H_2 & = & \ms H & \qquad \cong \qquad & \ordind(\ms H, f(1) \ldots, f(m)) 
\end{array}
$$

For all $1 \le j \le m$, it holds
$$
f(j) = \begin{cases}
	h_1(c_j) & \text{if } c_j \in A_1, \\
	h_2(c_j) & \text{if } c_j \in A_2,
\end{cases}
$$
where $f\colon [m] \to H_1 \uplus H_2$ is the map from Definition~\ref{def:moreotimes}.
If $c_j \in A_i$, then $c_j = c_{i,k}$ for some $1 \le k \le m_i$ and therefore $d_j = (i, k)$.  This implies
$h_i(c_j) = [k]_{\bar c_i} = f(j)$ by Definition~\ref{def:ordind} and Definition~\ref{def:treecross}.
Therefore, $f(j_1) = f(j_2)$ iff $c_{j_1} = c_{j_2}$, which then implies lemma.
\qed
\end{proof}

We now define the \emph{tree cross product} $R_1 \mycart{q}{\otimes}{\bar d} R_2$ of $R_1$ and~$R_2$
and then show that in fact $R = R_1 \mycart{q}{\otimes}{\bar d} R_2$.
As motivated before, the root of the tree cross product is simply $\root(R_1) \moreotimes{\bar d} \root(R_2)$.
For the construction of the subtrees, recall that each subtree of~$R$ corresponds to 
either a set move $U \subseteq A$ or a point move $a \in A$.
Here, $\{\, U \subseteq A\,\} = \{\,U_1 \uplus U_2 \mid U_1 \subseteq A_1, U_2 \subseteq A_2\,\}$ and
$A = A_1 \uplus A_2$.
We can therefore reconstruct the subtrees of $R$ by recursively combining each 
subtree for a set $U_1 \subseteq A_1$ with
a subtree for a set $U_2 \subseteq A_2$ (the set $S_2$ in the following definition), 
and by choosing subtrees of $R_1$ for point moves
in $A_1$, and choosing subtrees of $R_2$ for point moves in~$A_2$ 
(the set $S_1$ in the following definition).

\begin{definition}[Tree Cross Product]
\label{def:treecross}
\rm
Let~$\ms A_1$ and $\ms A_2$ be $\tau$-structures and 
let $\otimes = {\myotimes{g}{f_1}{f_2}}$ for 
 some $t$-relabelings $g, f_1, f_2$.
For a fixed constant~$q \ge 0$, let~$m$ and~$p$ be nonnegative integers 
such that~$m + p \le q$. Let $\bar c = c_1,\ldots,c_m \in (A_1 \cup A_2)^m$
and $\bar C = C_1, \ldots, C_p$, where $C_j \subseteq A_1 \cup A_2$, $1 \le j \le p$.
For $i \in \{1,2\}$, let $\bar c_i = c_{i,1},\ldots,c_{i,m_i} = \bar c[A_i]$, 
$q_i \ge q - m - p$, and $R_i = \redchar{q_i}{\ms A_i}{\bar c_i}{\bar C \cap A_i}$
with $\root(R_i) = (\ms H_i, \bar c_i', \bar C_i') = \ordind(A_i, \bar c_i, \bar C \cap A_i)$.
We define the \bem{tree cross product} of $R_1$ and $R_2$, 
$R = R_1 \mycart{q}{\otimes}{\bar d} R_2$, be a finite, rooted tree such that
\begin{itemize}
\item $\root(R) = \root(R_1) \moreotimes{\bar d} \root(R_2)$, and
\item if $m + p + 1 \le q$, then $\subtrees(R) = S_1 \cup S_2$, where
\begin{align*}
S_1 = {}
	& \big\{\, \subtree_{R_1}(u_1) \mycart{q}{\otimes}{\bar d(1,m_1+1)} R_2 \bigm|  \\
	& \qquad \qquad u_1 = (\ms H_1', \bar c_1' c, \bar C_1') \in \children_{R_1}(\root(R_1)) \,\big\} \cup {} \\
	& \big\{\, R_1 \mycart{q}{\otimes}{\bar d (2,m_2+1)} \subtree_{R_2}(u_2) \bigm|  \\
	& \qquad \qquad
	 	u_2 = (\ms H_2', \bar c_2' c, \bar C_2') \in \children_{R_2}(\root(R_2)) \,\big\}
\end{align*}
and
\begin{multline*}
S_2 = \big\{\, \subtree_{R_1}(u_1) \mycart{q}{\otimes}{\bar d} \subtree_{R_2}(u_2) \bigm| \\
	u_i = (\ms H_i', \bar c_i', \bar C_i' D_i) \in \children_{R_i}(\root(R_i)), 1 \le i \le 2 \,\big\}.
\end{multline*}
\end{itemize}
\end{definition}

\begin{lemma}\label{lem:tree_cross_product}
Let~$\ms A_1$ and $\ms A_2$ be $\tau$-structures and 
let $\otimes = \myotimes{g}{f_1}{f_2}$ for some
 $t$-relabelings $g, f_1, f_2$. For nonnegative integers~$q,m,p$
 with~$m + p \leq q$, let $\bar c = c_1, \ldots, c_m \in (A_1 \cup A_2)^m$
and $\bar C = C_1, \ldots, C_p$, where $C_j \subseteq A_1 \cup A_2$ for 
$1 \le j \le p$. Also let $\bar{d} = \indicator(A_1,A_2,\bar{c})$ 
and for~$1 \leq i \le 2$ let~$q_i \geq q - m - p$. Then
%Using notation from Definition~\ref{def:treecross}, we get
$$
\redchar{q}{\ms A_1 \otimes \ms A_2}{\bar c}{\bar C} =
\redchar{q_1}{\ms A_1}{\bar c_1}{\bar C \cap A_1}
\mycart{q}{\otimes}{\bar d}
\redchar{q_2}{\ms A_2}{\bar c_2}{\bar C \cap A_2}.
$$
\end{lemma}
\begin{proof}
The proof is an induction over $q - m - p$.
By Lemma~\ref{lem:moreotimes}, 
$$
\root(\redchar{q}{\ms A_1 \otimes \ms A_2}{\bar c}{\bar C}) = \root(R_1) \moreotimes{\bar d} \root(R_2).
$$
If $q - m - p = 0$, then $\redchar{q}{\ms A_1 \otimes \ms A_2}{\bar c}{\bar C}$
consists of a single root node and the lemma holds.
%Otherwise, let $S_1$ and $S_2$ be defined as in Definition~\ref{def:treecross}.
Otherwise, the set of subtrees is by definition
\begin{align*}
	\subtrees(\redchar{q}{\ms{A}}{\bar{c}}{\bar{C}}) = {} 
	& \big\{\, \redchar{q}{\ms{A}}{\bar{c} d}{\bar{C}} \mid d \in A \,\big\} \cup {} \\
	& \big\{\, \redchar{q}{\ms{A}}{\bar{c}}{\bar{C} D} \mid D \subseteq A \,\big\}.
\end{align*}
Here, by the induction hypothesis 
\begin{align*}
	& \big\{\, \redchar{q}{\ms{A}}{\bar{c} d}{\bar{C}} \mid d \in A \,\big\} \\
= {} 
	& \big\{\, \redchar{q}{\ms{A}}{\bar{c} d}{\bar{C}} \mid d \in A_1 \,\big\} \cup 
	  \big\{\, \redchar{q}{\ms{A}}{\bar{c} d}{\bar{C}} \mid d \in A_2 \,\big\} \\
\stackrel{\text{i.h.}}= {}
	& \big\{\, \redchar{q}{\ms A_1}{\bar{c}[A_1] d}{\bar{C} \cap A_1} \mycart{q}{\otimes}{\bar d(1,m_1+1)} R_2 \mid d \in A_1 \,\big\} \cup {} \\
	& \big\{\, R_1 \mycart{q}{\otimes}{\bar d(2,m_2+1)} \redchar{q}{\ms A_2}{\bar{c}[A_2] d}{\bar{C} \cap A_2} \mid d \in A_2 \,\big\}  \\
= {} & S_1
\end{align*}
and, similarly,
\begin{align*}
	& \big\{\, \redchar{q}{\ms{A}}{\bar{c}}{\bar{C} D} \mid D \subseteq A \,\big\} \\
\stackrel{\text{i.h.}}= {}
	& \big\{\,
		\redchar{q}{\ms A_1}{\bar c[A_1]}{\bar C D \cap A_1} 
		\mycart{q}{\otimes}{\bar d}
		\redchar{q}{\ms A_2}{\bar c[A_2]}{\bar C D \cap A_2} 
		\bigm| \\
	& \qquad \qquad D \in U \,\big\} \\
= {}	& S_2.
\end{align*}

This concludes the proof.
\qed
\end{proof}

\begin{lemma}\label{lemma:combination_time}
Given $R_1$ and $R_2$, the tree cross product $R_1 \mycart{q}{\otimes}{\bar d} R_2$ can be computed time
${\it poly}(|R_1|,|R_2|)$, where $|R_i|$ denotes the number of nodes in $R_i$.
\end{lemma}
\begin{proof}
An algorithm computing $R_1 \mycart{q}{\otimes}{\bar d} R_2$ may recursively traverse both trees top-down.
For each pair of subtrees $R_1'$ and $R_2'$ of $R_1$ and $R_2$, the algorithm has to be called only
once.  The number of recursive calls is therefore bounded by $|R_1|\cdot|R_2|$ and each recursive call
takes time dependent on~$q$ and $\tau$ only.
\qed
\end{proof}

We can now finally prove the Main Theorem.
\themainthm
\begin{proof}
It is no 
loss of generality to assume that~$G$ has at least one vertex. Otherwise
deciding whether~$G \models \varphi$ takes constant time.
By Lemmas~\ref{lemma:ef_model_checking}, \ref{lemma:games_full_char_trees} 
and~\ref{lemma:games_red_char_trees}, to prove that~$G \models \varphi$
it is sufficient to show that the verifier has a winning strategy in the 
model checking game $\mcg(\rc{q}{G},\varphi,\epsilon,\epsilon)$.
By Lemma~\ref{lemma:red_char_size}, the size of the reduced characteristic 
tree~$\rc{q}{G}$ of a $t$-labeled graph is at most~$f_1(q,t)$ for some
computable function~$f_1$ of~$q$ and~$t$ alone. By Lemma~\ref{lemma:combination_time},
the time taken to combine two reduced characteristic trees of size~$f_1(q,t)$ 
is~$f(q,t) = \mbox{poly}(f_1(q,t))$. 

We claim that the total time taken to construct~$\rc{q}{G}$ from its parse
tree decomposition~$T$ is $O(f(q,t) \cdot |T|)$. The proof is by induction 
on~$|T|$. By Lemma~\ref{lemma:construction_constant_size}, the claim
holds when~$|T| = 1$. Suppose that $\bar{G} = \bar{G}_1 \myotimes{g}{h_1}{h_2} \bar{G}_2$,
where~$g,h_1,h_2$ are $t$-relabelings and let~$T_1$ and~$T_2$ be parse trees of~$\bar{G}_1$
and~$\bar{G}_2$, respectively. Then $|T| = |T_1| + |T_2| + 1$, where~$T$ is a parse tree of~$\bar{G}$. 
By induction hypothesis, one can construct the reduced characteristic trees~$\rc{q}{G_1}$ 
and~$\rc{q}{G_2}$ in times~$O(f(q,t) \cdot |T_1|)$ and~$O(f(q,t) \cdot |T_2|)$, respectively. 
By Lemma~\ref{lem:tree_cross_product}, one can indeed construct $\rc{q}{G}$
given~$\rc{q}{G_1}$, $\rc{q}{G_2}$ and~$\bar{d} = \varepsilon$.
By using Lemma~\ref{lemma:combination_time}, the time taken to construct~$\rc{q}{G}$ is
\[O(f(q,t) + f(q,t) \cdot |T_1| + f(q,t) \cdot |T_2|) = O(f(q,t) \cdot |T|),\]
thereby proving the claim.

In order to check whether the verifier has a winning strategy in the model
checking game $\mcg(\rc{q}{G}, \varphi, \epsilon, \epsilon)$, one can use a
very simple recursive algorithm (see also~\cite{Gra07}).   A position
$p = (\psi, \bar x, \bar X, u)$ of the model checking game can
be identified with a call of the algorithm with arguments~$p$.
If $\psi$ is universal, then the algorithm recursively
checks whether the verifier has a winning strategy from all positions
$u'$ that are reachable from $u$ in the model checking game.
If otherwise $\psi$ is existential, then the algorithm checks whether there is one subsequent position
in the game from which the verifier has a winning strategy.
This algorithm visits each node
of the reduced characteristic tree~$\rc{q}{G}$ at most once.  Therefore the 
time taken to decide whether~$G \models \varphi$ is 
$O(f_1(q,t) + f(q,t) \cdot |T|) = O(f(q,t) \cdot |T|)$, as claimed. 
\qed
\end{proof}

\section{Discussion and Conclusion}
\label{sec:conclusion}
The proof of the Main Theorem shows that deciding whether a graph models
an \msoone-sentence is linear-time doable if the rankwidth of the graph is 
bounded. The theorem by Courcelle et al.~\cite{CMR00} says something stronger: 
one can compute the \emph{optimal} solution to a linear optimization problem expressible 
in \msoone\ in linear time for graphs of bounded rankwidth. In its simplest form, 
a \emph{linear optimization
problem} in \msoone\  is a tuple 
$$(\varphi(X_1, \ldots, X_l), a_1, \ldots, a_l, \opt),$$
where~$\varphi(X_1, \ldots, X_l)$ is an \msoone-formula with
the free set variables~$X_1, \ldots, X_l$, $\bar{a} = a_1, \ldots, a_l \in \mathbf{Z}^l$,
and~$\opt$ is either~$\max$ or~$\min$. 
The objective is, given an input graph~$G=(V,E)$, to find~$(U_1, \ldots, U_l) \subseteq V^l$ 
such that~$G \models \varphi[X_1/U_1, \ldots, X_l/U_l]$ and~$\sum_{i=1}^l a_i |U_i|$
is optimized (maximized or minimized).

One can use the techniques outlined in this paper to prove the stronger statement by first
constructing reduced characteristic trees $\redchar{q}{G}{\varepsilon}{U_1, \ldots, U_l}$,
of which there are only a function of~$q$ and~$l$. All that remains to do is simulate
the model checking game on each of the reduced characteristic trees and output the tuple 
$(U_1, \ldots, U_l)$ for which there is a winning strategy and $\sum_{i=1}^l a_i |U_i|$
is optimized. 

Moreover the results of this paper naturally extend to directed graphs and birankwidth. 
This allows us to conclude that any decision or optimization problem on directed 
graphs expressible in \msoone\ is linear-time solvable on graphs of bounded 
birankwidth~\cite{CMR00,Kan07}. Finally, the game-theoretic approach has already
been used to prove Courcelle's result for treewidth~\cite{Cou90,ALS91,CM93}
with an emphasis on practical implementability~\cite{KLR10}.

%\bibliographystyle{abbrv}
%\bibliography{cross,conf}

\begin{thebibliography}{10}

\bibitem{ALS91}
S.~Arnborg, J.~Lagergren, and D.~Seese.
\newblock Easy problems for tree-decomposable graphs.
\newblock {\em J. Algorithms}, 12(2):308--340, 1991.

\bibitem{BCG82}
E.~R. Berlekamp, J.~H. Conway, and R.~K. Guy.
\newblock {\em Winning Ways for Your Mathematical Plays}.
\newblock A.K. Peters, 1982.

\bibitem{Cou90}
B.~Courcelle.
\newblock The monadic second order theory of {G}raphs {I}: {R}ecognisable sets
  of finite graphs.
\newblock {\em Information and Computation}, 85:12--75, 1990.

\bibitem{Cou94}
B.~Courcelle.
\newblock Monadic second-order definable graph transductions: A survey.
\newblock {\em Theor. Comput. Sci.}, 126(1):53--75, 1994.

\bibitem{ck07}
B.~Courcelle and M.~M. Kante.
\newblock Graph operations characterizing rank-width and balanced graph
  expressions.
\newblock In {\em Proc.\ of 33rd WG}, number 4769 in LNCS, pages 66--75. Springer,
  2007.

\bibitem{CMR00}
B.~Courcelle, J.~A. Makowsky, and U.~Rotics.
\newblock {Linear Time Solvable Optimization Problems on Graphs of Bounded
  Clique Width}.
\newblock {\em Theory of Computing Systems}, 33:125--150, 2000.

\bibitem{CMR01}
B.~Courcelle, J.~A. Makowsky, and U.~Rotics.
\newblock On the fixed parameter complexity of graph enumeration problems
  definable in monadic second-order logic.
\newblock {\em Discrete Applied Mathematics}, 108(1-2):23--52, 2001.

\bibitem{CM93}
B.~Courcelle and M.~Mosbah.
\newblock Monadic second-order evaluations on tree-decomposable graphs.
\newblock {\em Theor. Comput. Sci.}, 109(1-2):49--82, 1993.

\bibitem{EF99}
H.-D. Ebbinghaus and J.~Flum.
\newblock {\em Finite Model Theory}.
\newblock Springer, 1999.

\bibitem{FV59}
S.~Feferman and R.~Vaught.
\newblock The first order properties of algebraic systems.
\newblock {\em Fund. Math}, 47:57--103, 1959.

\bibitem{GH10}
R.~Ganian and P.~Hlin\v{e}en{\'y}.
\newblock On parse trees and {M}yhill--{N}erode--type tools for handling graphs
  of bounded rank-width.
\newblock {\em Discrete Applied Mathematics}, 158(7):851--867, 2010.

\bibitem{GHO09}
R.~Ganian, P.~Hlin\v{e}n\'y, and J.~Obdr\v{z}\'alek.
\newblock Unified approach to polynomial algorithms on graphs of~bounded
  (bi-)rank-width.
\newblock Submitted, 2009.

\bibitem{Gra07}
E.~Gr{\"a}del.
\newblock Finite model theory and descriptive complexity.
\newblock In {\em Finite Model Theory and Its Applications}, pages 125--230.
  Springer, 2007.

\bibitem{Gur79}
Y.~Gurevich.
\newblock {Modest Theory of Short Chains. I}.
\newblock {\em J. Symb. Log.}, 44(4):481--490, 1979.

\bibitem{Gur85}
Y.~Gurevich.
\newblock Monadic second-order theories.
\newblock In S.~F. Jon~Barwise, editor, {\em Model-Theoretic Logics}, pages
  479--506. Springer-Verlag, 1985.

\bibitem{Hin73}
J.~Hintikka.
\newblock {\em Logic, Language-Games and Information: Kantian Themes in the
  Philosophy of Logic}.
\newblock Clarendon Press, 1973.

\bibitem{HO08}
P.~Hlin\v{e}n\'y and S.~Oum.
\newblock Finding branch-decomposition and rank-decomposition.
\newblock {\em SIAM Journal on Computing}, 38:1012--1032, 2008.

\bibitem{Kan07}
M.~M. Kante.
\newblock The rankwidth of directed graphs.
\newblock Preprint. Available at: \\ \texttt{http://arxiv.org/abs/0709.1433},
  2007.

\bibitem{KLR10}
J.~Kneis, A.~Langer, and P.~Rossmanith.
\newblock Courcelle's {T}heorem -- a game-theoretic approach, 2010.
\newblock submitted.

\bibitem{Oum05}
S.~Oum.
\newblock {\em Graphs of Bounded Rankwidth}.
\newblock PhD thesis, Princeton University, 2005.

\bibitem{OS06b}
S.~Oum and P.~D. Seymour.
\newblock Approximating clique-width and branch-width.
\newblock {\em Journal of Combinatorial Theory Series B}, 96(4):514--528, 2006.

\bibitem{os06}
L.~{\O}verlier and P.~Syverson.
\newblock Locating hidden servers.
\newblock In {\em Proceedings of the 2006 IEEE Symposium on Security and
  Privacy}. IEEE CS, May 2006.

\end{thebibliography}

\end{document}